\documentclass[journal,twoside,web]{ieeecolor}

\usepackage{generic}
\usepackage{cite}
\usepackage{amsmath,amssymb,amsfonts}
\usepackage{algorithmic}
\usepackage{graphicx}
\usepackage{textcomp}

\usepackage{graphicx}
\usepackage{subfigure}
\usepackage{epsfig} 
\usepackage{cite}
\usepackage{lcsys}

\usepackage{bm}

\usepackage{cleveref}
\crefformat{equation}{(#2#1#3)}
\Crefname{algocfline}{Algorithm}{Algorithms}
\Crefname{algocf}{line}{lines}
\Crefname{AlgoLine}{Line}{Lines}
\crefname{AlgoLine}{line}{lines}

\newtheorem{definition}{Definition}

\newtheorem{corollary}{Corollary}
\newtheorem{remark}{Remark}
\newtheorem{proposition}{Proposition}

\begin{document}
\bstctlcite{BSTcontrol}

\def\BibTeX{{\rm B\kern-.05em{\sc i\kern-.025em b}\kern-.08em
    T\kern-.1667em\lower.7ex\hbox{E}\kern-.125emX}}
\markboth{\journalname, VOL. XX, NO. XX, XXXX 2017}
{Author \MakeLowercase{\textit{et al.}}: Preparation of Papers for IEEE Control Systems Letters (August 2022)}

\title{Chance-Constrained Correlated Equilibria \\ for Robust Noncooperative Coordination}

\author{
Jaehan Im$^{a}$,
Ufuk Topcu$^{b}$, and
David Fridovich-Keil$^{a}$
\thanks{This work was supported by the NSF under grants 2336840 and 2211548, by NASA under ULI grants 80NSSC21M0071 and 80NSSC24M0070, and by ONR under grant N00014-22-1-2703.}
\thanks{$^{a}$ Department of Aerospace Engineering and Engineering Mechanics, The University of Texas at Austin; \texttt{jaehan.im@utexas.edu; dfk@utexas.edu}}
\thanks{$^{b}$  The Oden Institute for Computational Engineering and Sciences, The University of Texas at Austin; \texttt{utopcu@utexas.edu}}
}

\maketitle
\thispagestyle{empty}

\newcommand{\actSet}{\mathcal X}
\newcommand{\sys}{\text{sys}}

\newcommand{\ones}{\bm 1}
\newcommand{\reals}{{\mbox{\bf R}}}
\newcommand{\integers}{{\mbox{\bf Z}}}
\newcommand{\symm}{{\mbox{\bf S}}}  
\newcommand{\lag}{\mathcal{L}}

\newcommand{\nullspace}{{\mathcal N}}
\newcommand{\range}{{\mathcal R}}
\newcommand{\Rank}{\mathop{\bf Rank}}
\newcommand{\Tr}{\mathop{\bf Tr}}
\newcommand{\diag}{\mathop{\bf diag}}
\newcommand{\card}{\mathop{\bf card}}
\newcommand{\rank}{\mathop{\bf rank}}
\newcommand{\conv}{\mathop{\bf conv}}
\newcommand{\prox}{\bm{prox}}

\newcommand{\Expect}{\mathop{\bf E{}}}
\newcommand{\Prob}{\mathop{\bf Prob}}
\newcommand{\Co}{{\mathop {\bf Co}}} 
\newcommand{\dist}{\mathop{\bf dist{}}}
\newcommand{\argmin}{\mathop{\rm argmin}}
\newcommand{\argmax}{\mathop{\rm argmax}}
\newcommand{\epi}{\mathop{\bf epi}} 
\newcommand{\Vol}{\mathop{\bf vol}}
\newcommand{\dom}{\mathop{\bf dom}} 
\newcommand{\intr}{\mathop{\bf int}}
\newcommand{\sign}{\mathop{\bf sign}}
\newcommand{\norm}[1]{\left\lVert#1\right\rVert}
\newcommand{\mnorm}[1]{{\left\vert\kern-0.25ex\left\vert\kern-0.25ex\left\vert #1 
    \right\vert\kern-0.25ex\right\vert\kern-0.25ex\right\vert}}

\newcommand{\cf}{{\it cf.}}
\newcommand{\eg}{{\it e.g.}}
\newcommand{\ie}{{\it i.e.}}
\newcommand{\etc}{{\it etc.}}

\newcommand{\ba}[2][]{\todo[color=orange!40,size=\footnotesize,#1]{[BA] #2}}

\newcommand{\fix}[1]{\textcolor{red}{#1}}

\newcommand{\bigO}{\mathcal{O}}

\newcommand{\intSet}{\mathbb{Z}}
\newcommand{\realSet}{\mathbb{R}}
\newcommand{\natSet}{\mathbb{N}}
\newcommand{\zeroSet}{\bm{0}}
\newcommand{\state}{\bm{x}}
\newcommand{\cmdh}{\bar{\bm{u}}}
\newcommand{\cmda}{\mathring{\bm{u}}}
\newcommand{\cmd}{\bm{u}}
\newcommand{\costh}{J_h}
\newcommand{\costa}{J_a}
\newcommand{\observ}{\bm{z}}

\begin{abstract}
Correlated equilibria enable a coordinator to influence the self-interested agents by recommending actions that no player has an incentive to deviate from.
However, the effectiveness of this mechanism relies on accurate knowledge of the agents’ cost structures. 
When cost parameters are uncertain, the recommended actions may no longer be incentive compatible, allowing agents to benefit from deviating from them.
We study a chance-constrained correlated equilibrium problem formulation that accounts for uncertainty in agents' costs and guarantees incentive compatibility with a prescribed confidence level. 
We derive sensitivity results that quantify how uncertainty in individual incentive constraints affects the expected coordination outcome. 
In particular, the analysis characterizes the value of information by relating the marginal benefit of reducing uncertainty to the dual sensitivities of the incentive constraints, providing guidance on which sources of uncertainty should be prioritized for information acquisition. 
The results further reveal that increasing the confidence level is not always beneficial and can introduce a tradeoff between robustness and system efficiency. 
Numerical experiments demonstrate this tradeoff: CC-CE reduces realized coordination cost by up to 35\% at intermediate confidence levels, while the proposed information-gain metric consistently identifies effective uncertainty sources to reduce.
\end{abstract}

\begin{IEEEkeywords}
Game theory, Noncooperative coordination, Uncertain systems, Optimization
\end{IEEEkeywords}

\section{Introduction}
\IEEEPARstart{M}{ulti-agent} systems often consist of self-interested agents whose decisions collectively determine system-level outcomes. 
In such environments, agents optimize their own objectives independently, which often leads to inefficient outcomes at equilibrium. As a result, designing coordination mechanisms that improve system-level performance without requiring cooperation remains a central challenge in decentralized decision-making systems.

The correlated equilibrium concept \cite{aumann, aumann_2} provides a coordination mechanism that can improve outcomes in noncooperative games. In this mechanism, a coordinator samples a joint action from a probability distribution and privately recommends an action to each player. If the distribution satisfies the correlated equilibrium conditions, no player can reduce its cost by unilaterally deviating from the recommendation. This incentive-compatibility property enables a coordinator to influence the behavior of self-interested agents while respecting their autonomy. The correlated equilibrium concept has been widely studied as a coordination tool in various multi-agent systems \cite{rrce, ce_effective_traffic, ce_effective_experimental, j_vq, ce_error_2_exp}, as illustrated in the air mobility coordination example in \Cref{fig:concept}.

However, in many practical applications, the coordinator may not have accurate knowledge of the players' cost structures \cite{ce_error_1_belief, ce_error_2_exp, ce_error_3_inequity,ce_error_4_payoffSensitive}. 
Under such uncertainty, the incentive compatibility conditions derived from the nominal cost model may no longer hold. Consequently, players may have incentives to deviate from the recommended actions, which undermines the effectiveness of the correlated equilibrium.

\begin{figure}[t]
    \centering
    \includegraphics[width=0.99\linewidth]{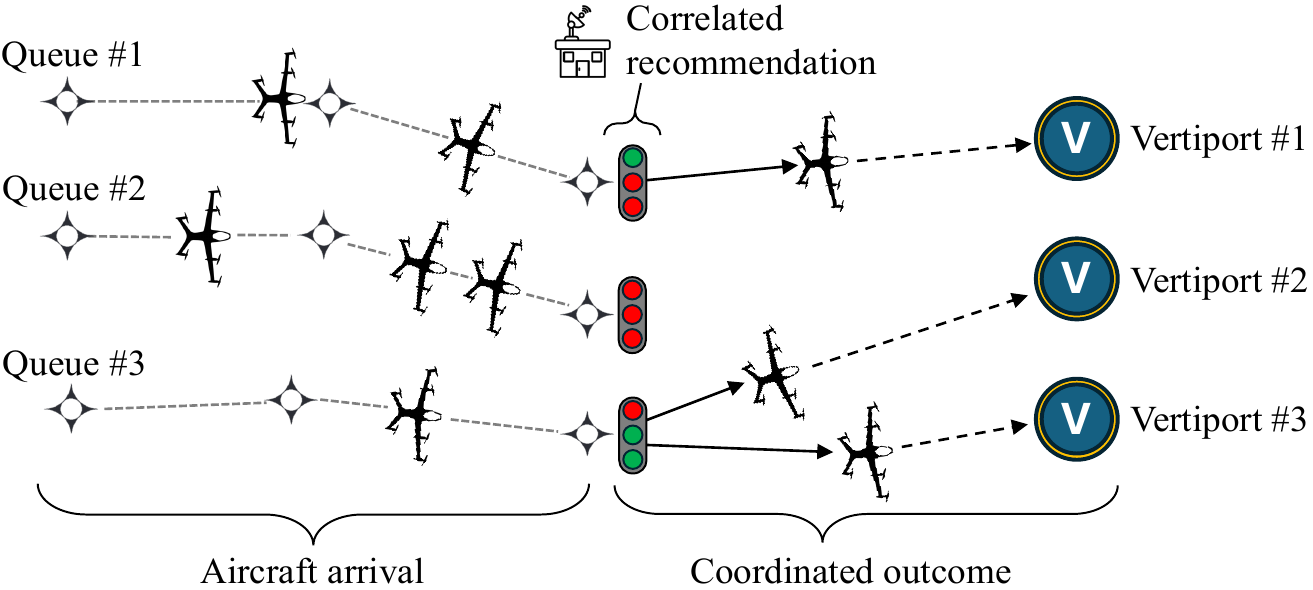}
    \caption{
    Illustration of the vertiport occupancy coordination scenario.
    Aircraft arrive through multiple approach queues, where each queue represents a self-interested agent that dispatches aircraft and determines which combination of vertiports to occupy or yield.
    A central coordinator broadcasts correlated recommendations that suggest vertiport usage while preserving agents' autonomy.
    The interaction forms a noncooperative coordination game in which agents may follow or deviate from the recommended actions.
    }
    \label{fig:concept}
\end{figure}

We study a \emph{chance-constrained correlated equilibrium} (CC-CE) that explicitly accounts for uncertainty in agents' costs.
In CC-CE, the incentive compatibility conditions are required to hold with a user-specified confidence level.
By incorporating uncertainty directly into the equilibrium constraints, CC-CE generates recommendations that remain incentive-compatible under cost uncertainty with quantifiable confidence.
Our recent work has explored the use of CC-CE for coordination under uncertainty \cite{j_vq}, primarily focusing on empirical performance. However, the theoretical role of uncertainty in shaping coordination outcomes and incentive constraints remains largely unexplored.

The contributions of this paper are threefold. First, we derive sensitivity results that quantify the value of information in individual incentive constraints and identify the constraints that act as information bottlenecks in coordination.
Second, we analyze how the confidence level used in CC-CE affects system performance, revealing a tradeoff between robustness and efficiency.
Third, we provide numerical experiments that illustrate how these theoretical insights explain the coordination behavior observed in practice.

\section{Related work}

The correlated equilibrium concept \cite{aumann, aumann_2} has been widely studied as a mechanism for coordinating noncooperative agents, as it allows a coordinator to recommend actions while preserving each agent's autonomy \cite{rrce, ce_effective_traffic, ce_effective_experimental, ce_error_2_exp, j_vq}. 
Correlated equilibria may also emerge through decentralized learning dynamics such as regret matching \cite{regret}, but we focus on equilibrium selection under explicit coordination objectives.
However, equilibrium recommendations can be sensitive to uncertainty in the underlying payoff structure, such as estimation errors or stochastic perturbations \cite{ce_effective_experimental, ce_effective_traffic}. 

Existing work addressing uncertainty in correlated equilibria includes Bayes correlated equilibria, which model payoff uncertainty through states and information structures and define incentive compatibility with respect to players' beliefs \cite{ce1,ce6,ce7}. 
Robust and distributionally robust formulations instead incorporate uncertainty directly into agents' cost functions and enforce guarantees over admissible uncertainty sets \cite{ce2,ce4,ce8,ce9}. 
While these approaches provide strong robustness, they can lead to nonconvex optimization problems and become conservative when uncertainty represents stochastic fluctuations around nominal costs \cite{ce8,ce9}.

Bounded-rationality models, such as quantal response equilibria and their correlated variants, introduce stochastic response rules and define equilibrium as a fixed point of noisy best responses \cite{qre1,qre2}. 
In contrast, our setting treats randomness as uncertainty in the deviation incentives themselves and imposes explicit reliability constraints on correlated-equilibrium recommendations.

\section{Preliminaries}

\subsection{Correlated equilibrium}

Consider a finite strategic game with agent set $N=\{1,\dots,n\}$. 
Each agent $i\in N$ chooses an action from a finite action set 
$X_i$. The joint action space is defined as
$\actSet := \prod_{i\in N} X_i$.
Let $z \in \Delta(\actSet)$ denote a joint probability distribution 
over action profiles. A correlated equilibrium can be interpreted 
as a coordination device that recommends actions to agents 
according to the distribution $z$. Upon receiving a recommendation, 
each agent decides whether to follow it or deviate to another action.

Let $J_i:\actSet \to \realSet$ denote the cost of agent $i$. 
We define the deviation cost for agent $i$ as
\begin{equation}\label{eq:devCost}
\Delta J_i(x_i,x_i',x_{-i})
:=
J_i(x_i,x_{-i})
-
J_i(x_i',x_{-i}),
\end{equation}
where $x_{-i} \in \prod_{j\neq i} X_j$ denotes the joint actions 
of all agents except $i$.
This quantity represents the cost change experienced by agent $i$ 
if the recommended action $x_i$ is replaced with a deviation $x_i'$ 
while other agents play $x_{-i}$.

Given a recommendation $x_i$,
we define the unconditional deviation margin as
\begin{equation}\label{eq:uncondDevMargin}
M_i(x_i,x_i';z)
:=
\sum_{x_{-i}} z(x_i,x_{-i})
\Delta J_i(x_i,x_i',x_{-i}).
\end{equation}
This quantity is the probability-weighted deviation margin, capturing whether the deviation is profitable on average under the recommendation distribution, and is linear in $z$.
We recall the classical definition of correlated equilibrium \cite{aumann,aumann_2}.

\begin{definition}[Correlated equilibrium]
A distribution $z \in \Delta(\actSet)$ is a correlated equilibrium if
\begin{equation} \label{eq:ceDef}
M_i(x_i,x_i';z)
\le 0,
\quad
\forall i,
\end{equation}
for all suggestions $x_i$, and for all deviations $x_i' \neq x_i$.
\end{definition}

Intuitively, once an agent receives a recommendation, no unilateral 
deviation yields a lower expected cost. Therefore following the 
recommendation is rational.

\subsection{Coordination via correlated equilibrium}

The objective of coordination is to select a feasible equilibrium that 
leads to a desirable system-level outcome. Let 
$J_\sys:\Delta(\actSet)\to\realSet$ denote a system-level cost function.

A correlated equilibrium specifies a probability distribution over 
joint actions, so we evaluate the system outcome in expectation. 
Consequently, many coordination objectives admit a 
linear combination of outcome costs with respect to the distribution $z$. 
For example, we can express the utilitarian system cost as
$J_\sys(z)=
\sum_{x\in\actSet} z(x)
\left(\sum_{i\in N} J_i(x)\right)$.

We formulate the coordination task as the following equilibrium 
selection problem
\begin{equation} \label{eq:coordProb}
\begin{aligned}
\min_{z \in \Delta(\actSet)} \quad & J_\sys(z) \\
\text{s.t.} \quad 
& z \text{ satisfies the constraints \Cref{eq:ceDef}}.
\end{aligned}
\end{equation}
Because both the objective and the constraints are linear in $z$, 
\Cref{eq:coordProb} becomes a linear program.

\section{Correlated equilibrium \\ under cost uncertainty}
We formalize the CC-CE coordination model that serves as the backbone for the subsequent sensitivity and information-acquisition analysis.

\subsection{Deviation uncertainty model}

In practical coordination settings, the coordinator may not know the 
exact deviation costs of agents due to modeling errors, operational 
variability, or incomplete information about agent preferences. 
To capture this  uncertainty, we introduce a stochastic 
model for the deviation cost.

Recall that $\Delta J_i(x_i,x_i',x_{-i})$ denotes the deviation cost defined in \Cref{eq:devCost}. From the coordinator's perspective, this quantity is not perfectly known and this uncertainty is modeled as
\begin{equation} \label{eq:noiseModel}
\Delta J_i(x_i,x_i',x_{-i})
=
\Delta \bar J_i(x_i,x_i',x_{-i})
+
\eta_i ,
\end{equation}
where $\Delta \bar J_i$ represents the nominal estimate of the 
deviation cost and $\eta_i$ is an agent-level uncertainty term. 
In this work we assume $\eta_i \sim \mathcal{N}(0,\sigma_i^2)$ for analytical tractability, where 
$\sigma_i \geq 0$ characterizes the uncertainty level associated with agent $i$.
The disturbance $\eta_i$ is assumed to be shared across all deviation-cost terms associated with agent $i$, i.e., across all $(x_i,x_i',x_{-i})$ combinations.
This captures agent-level modeling uncertainty that shifts all deviation margins for agent $i$ in a coherent way.

Substituting the noise model \Cref{eq:noiseModel} into the unconditional deviation margin \Cref{eq:uncondDevMargin}, for all $(x_i,x_i',x_{-i})$, yields
\begin{equation}
M_i(x_i,x_i';z)
=
\bar M_i(x_i,x_i';z)
+
\eta_i \pi_i(x_i),
\end{equation}
where $\bar M_i(x_i,x_i';z):=\sum_{x_{-i}} z(x_i,x_{-i})
\Delta \bar J_i(x_i,x_i',x_{-i})$ is the nominal deviation margin and $\pi_i(x_i) := \sum_{x_{-i}} z(x_i,x_{-i})$ is the recommendation probability
of action $x_i$.
Thus even if the nominal deviation margin is negative, 
uncertainty in the deviation cost may cause the realized margin to become positive, leading agents to deviate from the recommended action, cf. \Cref{eq:ceDef}.

\subsection{Chance-constrained correlated equilibrium}

To ensure robust coordination under uncertainty, we require the 
condition in \Cref{eq:ceDef} to hold with high confidence.
Building on our prior uncertainty-aware correlated-equilibrium study~\cite{j_vq}, we present a rigorous CC-CE formulation used throughout the remainder of the paper.

\begin{definition}[Chance-constrained correlated equilibrium]
For confidence level $\alpha \in (0,1)$, a distribution 
$z \in \Delta(\actSet)$ is a chance-constrained correlated equilibrium (CC-CE) if
\begin{equation} \label{eq:ccceDef}
\mathbb{P}
\Big(
\bar M_i(x_i,x_i';z)
+ \eta_i \pi_i(x_i) \le 0
\Big)
\ge \alpha
\quad
\forall i,
\end{equation}
for all recommendations $x_i$, and all deviations $x_i' \neq x_i$.
\end{definition}
\noindent
Each condition in \Cref{eq:ccceDef}, which we refer to as a \emph{deviation constraint} throughout the remainder of the paper, guarantees that deviations remain unprofitable with probability at least $\alpha$.

Under the noise model \Cref{eq:noiseModel} with $\eta_i \sim \mathcal{N}(0, \sigma_i^2)$, the deviation constraint \Cref{eq:ccceDef} admits 
a deterministic equivalent obtained via the standard Gaussian chance-constraint reformulation~\cite{deterministic}. Indexing each deviation constraint by $c=(i,x_i,x_i')$, and defining
$\pi_c := \pi_i(x_i)$, we write
\begin{equation}\label{eq:ccce_equiv}
\underbrace{
\bar M_i(x_i,x_i';z)
}_{=:\,m_c(z)}
+
\underbrace{\Phi^{-1}(\alpha)}_{=:\,q(\alpha)}\sigma_{i(c)} \pi_c
\le 0 ,
\end{equation}
where $i(c)$ denotes the agent associated with constraint $c$, and $\Phi^{-1}$ denotes the inverse cumulative distribution function of the normal distribution.
The term $m_c(z)$ denotes the nominal deviation margin, while
$q(\alpha)\sigma_{i(c)}\pi_c$ is the recommendation-weighted uncertainty margin.
This implies that the impact of uncertainty scales with the likelihood of the recommendation.
Since both $m_c(z)$ and $\pi_c$ are linear in $z$ for fixed $\alpha$ and $\sigma_{i(c)}$, each deviation constraint remains linear in $z$.

We obtain the uncertainty-aware coordination problem by replacing the nominal deviation constraints in \Cref{eq:ceDef} with the CC-CE deviation constraints \Cref{eq:ccce_equiv}. 
The resulting problem remains a linear program and serves as the foundation for the following analysis.

\section{Solution geometry and sensitivity analysis}

We analyze the feasible-set geometry of CC-CE and its sensitivity to uncertainty parameters.
The goal is to identify which deviation constraints in \Cref{eq:ccceDef} act as \emph{bottlenecks}, i.e., active deviation constraints that significantly restrict the achievable optimal coordination performance, and how reducing uncertainty can improve the outcome.
Concretely, we ask:
(i) which constraints become active, (ii) whether their restriction is more strongly shaped by the nominal game structure or by uncertainty in cost information, and (iii) which uncertainty sources are most valuable to reduce.

\begin{figure}[t!]
    \vspace{2mm}
    \centering
    \includegraphics[width=0.95\linewidth]{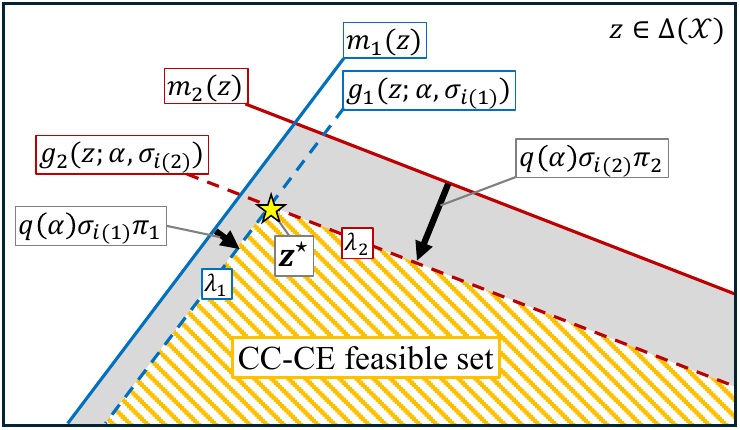}
    \caption{
    Conceptual illustration of CC-CE feasible-set contraction in the decision space $z \in \Delta(\mathcal{X})$ for two constraints $c \in \{1,2\}$. 
    Uncertainty contracts the nominal deviation constraints $m_c(z)$ by the weighted margin $q(\alpha)\sigma_{i(c)}\pi_c$, yielding the effective chance-constrained boundaries $g_c(z;\alpha,\sigma_{i(c)})$. 
    Constraint \(c=1\) (blue) has a relatively small uncertainty contribution, while constraint \(c=2\) (red) has a larger relative contribution from uncertainty-induced contraction.
    In this illustration, the optimal solution $z^\ast$ lies at the intersection where both constraints are active, with dual sensitivities $\lambda_1$ and $\lambda_2$. 
    This highlights that the value of reducing uncertainty depends jointly on the dual sensitivity, the uncertainty level, and recommendation likelihood.}
    \label{fig:geometry}
\end{figure}

\subsection{Bottleneck origin analysis} \label{sec:originAnalysis}

Recall that the CC-CE coordination problem is subject to
\begin{equation} \label{eq:constraints}
g_c(z;\alpha,\sigma_{i(c)})
=
m_c(z)
+
q(\alpha)\sigma_{i(c)}\pi_c
\le 0,
\quad \forall c,
\end{equation}
where $c=(i,x_i,x_i')$ indexes a deviation constraint and $i(c)$ denotes its associated agent.
Geometrically, the deviation constraints \Cref{eq:ccceDef} contract the feasible region defined by the nominal deviation constraints in \Cref{eq:ceDef} by the recommendation-weighted uncertainty margin $q(\alpha)\sigma_{i(c)}\pi_c$.
This contraction can make some deviation constraints active at the optimal coordination solution and thereby limit the achievable system performance.

Let $z^\star$ solve the CC-CE problem, and define the active set
$\mathcal{A}:=\{c \mid g_c(z^\star;\alpha,\sigma_{i(c)})=0\}$.
For $c\in\mathcal A$, the condition
$m_c(z^\star)+q(\alpha)\sigma_{i(c)}\pi_c=0$
shows how the active constraint balances the nominal deviation margin and the uncertainty-induced contraction, as illustrated in \Cref{fig:geometry}.
The margin $q(\alpha)\sigma_{i(c)}\pi_c$ indicates how strongly uncertainty contributes to the restriction imposed by the active constraint.
Large margins indicate that uncertainty in cost information contributes strongly to the active constraint, whereas small margins suggest that the restriction is more strongly shaped by the underlying nominal game structure.

\subsection{Information benefit: $\sigma$--sensitivity analysis} \label{sec:sigmaSensitivity}

While the previous subsection qualitatively interprets how nominal game structure and uncertainty contribute to active constraints, it does not quantify the benefit of reducing uncertainty.
We now analyze how the optimal system outcome changes with respect to the uncertainty levels $\sigma_i$.

\begin{proposition}[Local $\sigma$--sensitivity] \label{prop:sigma_sensitivity}
Let $J^\star_{\sys}(\sigma)$ denote the optimal system cost of the CC-CE coordination 
problem for a fixed confidence level $\alpha$, and let $\lambda_c^\star$ 
denote the optimal dual variable associated with constraint $c$. 
Then, locally around a parameter regime with fixed active constraints,
\begin{equation}
\textstyle{\frac{\partial J_{\sys}^\star}{\partial \sigma_i}
=
q(\alpha)\Lambda_i, \quad \forall i\in N},
\end{equation}
where
$\Lambda_i := \sum_{c\in C_i}\lambda_c^\star \pi_c$,
and $C_i$ denotes the set of deviation constraints associated with agent $i$.
\end{proposition}

\begin{proof}
For fixed $\alpha$, write the CC-CE constraints as
\begin{equation}
\textstyle{\min_{z\in\Delta(\actSet)} J_\sys(z)}
\quad
\text{s.t.}
\quad
m_c(z)+b_c(\sigma)\le 0,\ \forall c,
\end{equation}
where $b_c(\sigma)=q(\alpha)\sigma_{i(c)}\pi_c$.
Under a locally fixed active set, standard LP sensitivity results \cite{boyd} give
\begin{equation}
\textstyle{\frac{\partial J^\star_{\sys}}{\partial b_c} = \lambda_c^\star}.
\end{equation}
Applying the chain rule,
\begin{equation}
\begin{aligned}
\textstyle{\frac{\partial J_{\sys}^\star}{\partial \sigma_i}}
&=
\textstyle{\sum_c
\lambda_c^\star
q(\alpha)\pi_c\bm{1}\{i(c)=i\}}\\
&=
\textstyle{
q(\alpha)\sum_{c\in C_i}
\lambda_c^\star \pi_c}
=
q(\alpha)\Lambda_i .
\end{aligned}
\end{equation}
\end{proof}

\noindent
This result shows that, locally, reducing uncertainty for agent $i$ improves performance in proportion to the aggregate dual sensitivity of its associated deviation constraints.

\begin{corollary}[Constraint-level sensitivity] \label{cor:sigmaSensitivity}
The sensitivity of the optimal coordination outcome with respect to the uncertainty of constraint $c$ is
\begin{equation} \label{eq:sigmaSensitivity}
\textstyle{\frac{\partial J_{\sys}^\star}{\partial \sigma_{i(c)}}
=
q(\alpha)\lambda_c^\star \pi_c}.
\end{equation}
\begin{proof}
Immediate from \Cref{prop:sigma_sensitivity}.
\end{proof}
\end{corollary}

Thus, $\Lambda_i$ quantifies the aggregate value of reducing uncertainty
for agent $i$, whereas $\lambda_c^\star \pi_c$ captures the contribution of
an individual constraint.

\subsection{Information acquisition prioritization} \label{sec:informationTarget}

The marginal improvement in the optimal coordination outcome obtained by relaxing a constraint is proportional to the corresponding dual multiplier $\lambda_c$, as shown in \Cref{prop:sigma_sensitivity}.
Classically, $\lambda_c$ can be interpreted as the ``shadow price'' associated with constraint $c$.
However, this quantity alone does not determine which uncertainty source should be prioritized for information acquisition.

As shown in \Cref{eq:constraints}, each constraint margin consists of a structural component $m_c(z^\star)$ and a recommendation-weighted uncertainty component $q(\alpha)\sigma_{i(c)}\pi_c$.
Thus, even highly sensitive constraints may provide limited benefit from additional information if their uncertainty contribution is negligible.
Conversely, large uncertainty alone may have little effect when the corresponding constraint has low sensitivity.

To quantify the value of information, we combine the sensitivity result in \Cref{sec:sigmaSensitivity} with the uncertainty contribution discussed in \Cref{sec:originAnalysis}.
A first-order approximation of the improvement from eliminating $\sigma_{i(c)}$ is
\begin{equation}
\textstyle{
\Delta J_{\sys}^\star
\approx
-\sigma_{i(c)}
\frac{\partial J_{\sys}^\star}{\partial \sigma_{i(c)}}
=
-q(\alpha)\sigma_{i(c)}\lambda_c^\star \pi_c.}
\end{equation}
Since $q(\alpha)$ is constant for fixed confidence level, we approximate the potential information benefit of constraint $c$ by
\begin{equation} \label{eq:infogain}
\texttt{InfoGain}_c := \sigma_{i(c)} \lambda_c^\star \pi_c.
\end{equation}

This expression provides a first-order approximation of the marginal coordination benefit obtained by reducing uncertainty in constraint $c$.
Accordingly, effective information acquisition must jointly consider uncertainty magnitude, dual sensitivity, and recommendation probability.
If any of these terms is negligible, reducing uncertainty in that constraint may provide little coordination benefit.

\subsection{Confidence benefit: $\alpha$--$J_{\sys}$ tradeoff}

Finally, we view the confidence level $\alpha$ as a design knob.
Increasing $\alpha$ reduces deviation risk, but also tightens the CC-CE constraints and may increase coordination cost.

\begin{proposition}[Local $\alpha$ sensitivity] \label{prop:alpha_sensitivity}
Let $J_{\sys}^\star(\alpha)$ denote the optimal value of the CC-CE coordination
problem for fixed uncertainty levels $\{\sigma_i\}_{i\in N}$. Under Gaussian
uncertainty, the sensitivity of the optimal system cost locally around a fixed active set with respect to $\alpha$ is
\begin{equation}\label{eq:alpha_sensitivity}
\textstyle{
\frac{dJ_{\sys}^\star}{d\alpha}
=
\frac{1}{\phi(q(\alpha))}
\sum_i \Lambda_i \sigma_i},
\end{equation}
where $q(\alpha)=\Phi^{-1}(\alpha)$, $\phi$ is the standard normal density, and
$\Lambda_i= \sum_{c\in C_i}\lambda_c^\star \pi_c$ as defined in \Cref{prop:sigma_sensitivity}.
\end{proposition}

\begin{proof}
Recall that the CC-CE margin is
$b_c(\alpha,\sigma)=q(\alpha)\sigma_{i(c)}\pi_c$.
Using the local sensitivity result from \Cref{prop:sigma_sensitivity} and applying the chain rule,
\begin{equation}
\textstyle{
\frac{dJ_{\sys}^\star}{d\alpha}
=
\sum_c
\lambda_c^\star
\sigma_{i(c)}\pi_c
\frac{dq(\alpha)}{d\alpha}}.
\end{equation}
For Gaussian uncertainty,
\(
dq(\alpha)/d\alpha
=
1/\phi(q(\alpha))
\)
\cite{quantile}.
Grouping constraints by agent yields the result.
\end{proof}

\begin{corollary}[$\alpha$ sensitivity as information gain]
\label{cor:alpha_infogain}
Under Gaussian uncertainty,
\begin{equation} \label{eq:alpha_infogain}
\textstyle{
\frac{dJ_{\sys}^\star}{d\alpha}
=
\frac{1}{\phi(q(\alpha))}
\sum_c \texttt{InfoGain}_c},
\end{equation}
where $\texttt{InfoGain}_c$ is defined in \Cref{eq:infogain}.
\end{corollary}

\Cref{prop:alpha_sensitivity} and \Cref{cor:alpha_infogain} provide a local fixed active-set characterization showing that the $\alpha$-sensitivity of the optimal coordination cost scales with aggregate information gain.
Increasing $\alpha$ enlarges the recommendation-weighted uncertainty margins, shrinking the feasible coordination region and increasing coordination cost.
As $\alpha \to 1$, the factor $1/\phi(q(\alpha))$ grows rapidly, indicating high sensitivity under conservative confidence levels.

\begin{remark} [Confidence-level tradeoff] \label{rem:alpha_tradeoff}
Increasing $\alpha$ also reduces the probability of deviation.
To provide intuition regarding this tradeoff, consider the simplified effective cost
\begin{equation}
J_{\text{eff}}(\alpha)
=
J_{\sys}(\alpha)
+
(1-\alpha)C_{\text{dev}},
\end{equation}
where $C_{\text{dev}}$ denotes a constant surrogate penalty for deviation events.
Although deviation losses may generally depend on $z^\star(\alpha)$, we use a constant surrogate to illustrate the confidence-level tradeoff: the nominal outcome $J_{\sys}(\alpha)$ is realized when agents do not deviate, while deviation events occur with approximate residual probability $(1-\alpha)$ and incur penalty $C_{\text{dev}}$.
Using \Cref{eq:alpha_infogain}, a stationary tradeoff point satisfies
\begin{equation}\label{eq:alphaOpt}
\textstyle{
\frac{1}{\phi(q(\alpha))}
\sum_c \texttt{InfoGain}_c - C_{\text{dev}}}
= 0.
\end{equation}
Under this simplified surrogate model, increasing $\alpha$ has two competing effects: it reduces deviation risk through $(1-\alpha)C_{\text{dev}}$, but shrinks the feasible coordination region through the recommendation-weighted uncertainty margins.
Thus, increasing the confidence level does not always improve overall coordination performance.
\end{remark}

\section{Numerical experiment}

\subsection{Vertiport occupancy scenario}

We consider a vertiport occupancy coordination scenario for air mobility, illustrated in \Cref{fig:concept}.
Aircraft arrive through multiple queues, and each queue is modeled as a self-interested agent that determines which vertiports to occupy.
For each vertiport, the agent chooses whether to occupy or yield.

Let $x_i \in \{0,1\}^m$ denote the occupancy decision of agent $i$, where $m$ is the number of vertiports.
Here, $x_{iv}=1$ indicates that agent $i$ attempts to occupy vertiport $v$, and $x_{iv}=0$ indicates that it yields.
Thus, each agent has $2^m$ possible actions.
Let $N_v(x)=\sum_{i=1}^n x_{iv}$ denote the number of agents occupying vertiport $v$.
The nominal cost of each agent is
\begin{equation}
\small{
\textstyle{
\bar J_i(x)
=
5\gamma \sum_{v=1}^{m} x_{iv}(N_v(x)-1)
+
5 \sum_{v=1}^{m} (1-x_{iv})}},
\end{equation}
where $\gamma \in \realSet_{\geq1}$ controls the congestion penalty.
The first term penalizes simultaneous occupancy, while the second term penalizes yielding.

The coordinator selects a recommendation distribution $z$ over joint actions to minimize the aggregate expected delay
\begin{equation} \label{eq:scenario_jsys}
\textstyle{
J_{\sys}(z)
=
\mathbb{E}_{x\sim z}\!\left[\sum_{i=1}^{n} \bar J_i(x)\right]}.
\end{equation}

Uncertainty in the incentive constraints is modeled through Gaussian perturbations in the deviation costs, following \Cref{eq:noiseModel}.
Specifically, each agent's deviation cost $\Delta J_i(\cdot)$ includes an additive noise term $\eta_i \sim \mathcal{N}(0,\sigma_i^2)$.
The Gaussian uncertainty model represents aggregated operational uncertainty in practical air mobility coordination settings, including estimated time of arrival prediction errors, congestion-estimation uncertainty, and operational variability \cite{Gaussian_time, Gaussian_traj}.

\subsection{Baselines and evaluation metrics}

We compare \texttt{CC-CE} against \texttt{Naive-CE}, which solves the nominal correlated-equilibrium coordination problem in \Cref{eq:coordProb} without accounting for cost uncertainty.

We evaluate performance using three metrics: \textit{nominal system cost}, the value of \(J_{\sys}\) in \Cref{eq:scenario_jsys} under the recommendation distribution; \textit{realized system cost}, the aggregate delay after agents respond under realized perturbations; and \textit{deviation rate}, the fraction of sampled perturbation realizations with at least one profitable deviation.

\subsection{Experiment setup}

We conduct Monte Carlo experiments\footnote{The experiments were implemented in Julia v1.11.3.
We used \texttt{ParametricMCP} \cite{parammcp} and \texttt{PATHSolver} \cite{PATH} with default convergence tolerances to solve \Cref{eq:coordProb}. Code available at https://github.com/CLeARoboticsLab/CCCE} using 100 independent trials.
In each trial, we sample $\sigma_i \sim U(0,1.6)$ minutes for each agent.
For each computed recommendation distribution, we evaluate realized performance using 100 sampled realizations of the cost perturbations.

We consider $n=3$ queues and $m=2$ vertiports, so each agent has $2^m=4$ actions and the joint action space contains $(2^m)^n=64$ profiles.
We first compare \texttt{Naive-CE} with \texttt{CC-CE} across $\alpha \in \{0.66,0.8,0.9,0.95,0.99,0.999\}$ under $\gamma=2.0$.
This evaluates how confidence affects nominal cost, realized cost, and deviation rate.

We then evaluate information-acquisition strategies from \Cref{sec:informationTarget}.
After solving baseline \texttt{CC-CE}, we remove the uncertainty margins from five selected deviation constraints and resolve the coordination problem.
We compare random selection, $\sigma$-only selection, $\lambda$-only selection, and \texttt{InfoGain} selection using \Cref{eq:infogain}.
We repeat this experiment across multiple $\gamma$ values and report nominal cost normalized by the baseline \texttt{CC-CE} solution.
\begin{figure}[hbt!]
    \centering    \includegraphics[width=0.99\linewidth]{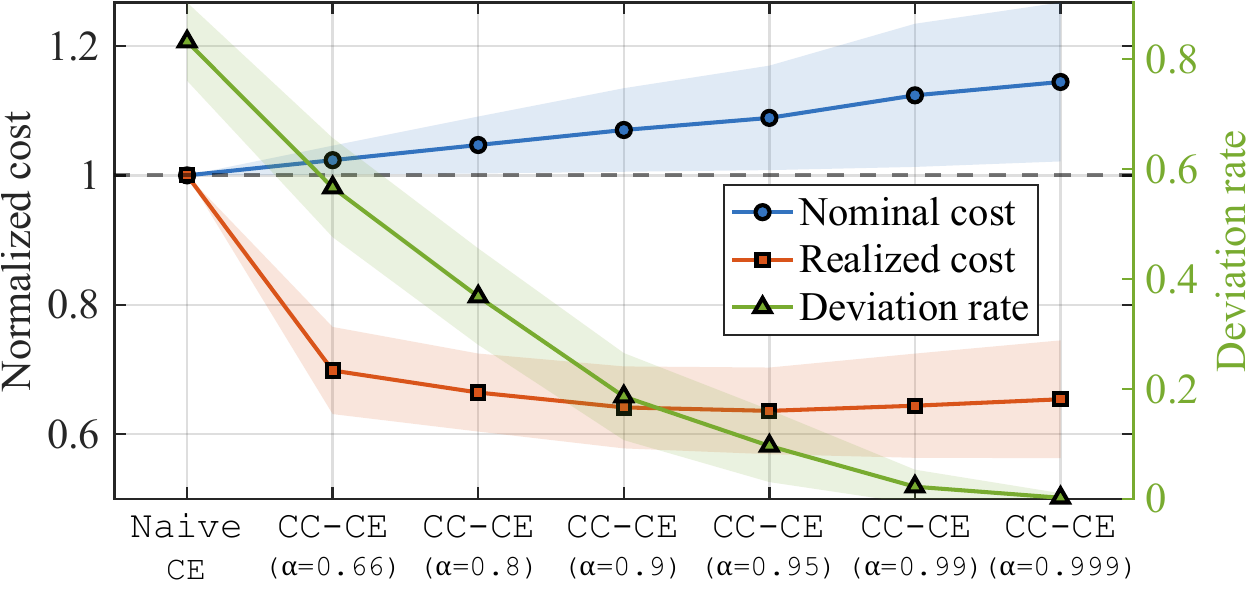}
    \caption{
    Effect of the confidence level $\alpha$ on coordination performance for the congestion parameter $\gamma=2.0$.
    Increasing $\alpha$ contracts the CC-CE feasible region and increases nominal cost, but reduces the deviation rate toward zero.
    The realized cost exhibits a non-monotonic trend, reflecting the effect of reduced deviations and increased conservativeness.
    Note that costs are normalized by their corresponding \texttt{Naive-CE} values, so nominal and realized costs should not be directly compared.
    }
    \label{fig:alphaTest}
\end{figure}

\subsection{Result: Robustness and confidence-level effects}

We first compare \texttt{Naive-CE} with \texttt{CC-CE} under different confidence levels.
As shown in \Cref{fig:alphaTest}, increasing $\alpha$ increases the nominal system cost due to contraction of the CC-CE feasible region, while the deviation rate decreases toward zero.

The realized system cost reflects the combined effect of these mechanisms.
Compared with \texttt{Naive-CE}, \texttt{CC-CE} reduces the realized system cost by 30--35\% at intermediate confidence levels.
However, overly large $\alpha$ eventually increases realized cost as feasible-set contraction dominates the benefit of reduced deviations.
This non-monotonic behavior is consistent with the confidence-level tradeoff in \Cref{rem:alpha_tradeoff}.

\begin{figure}[t!]
    \vspace{2mm}
    \centering
    \includegraphics[width=0.95\linewidth]{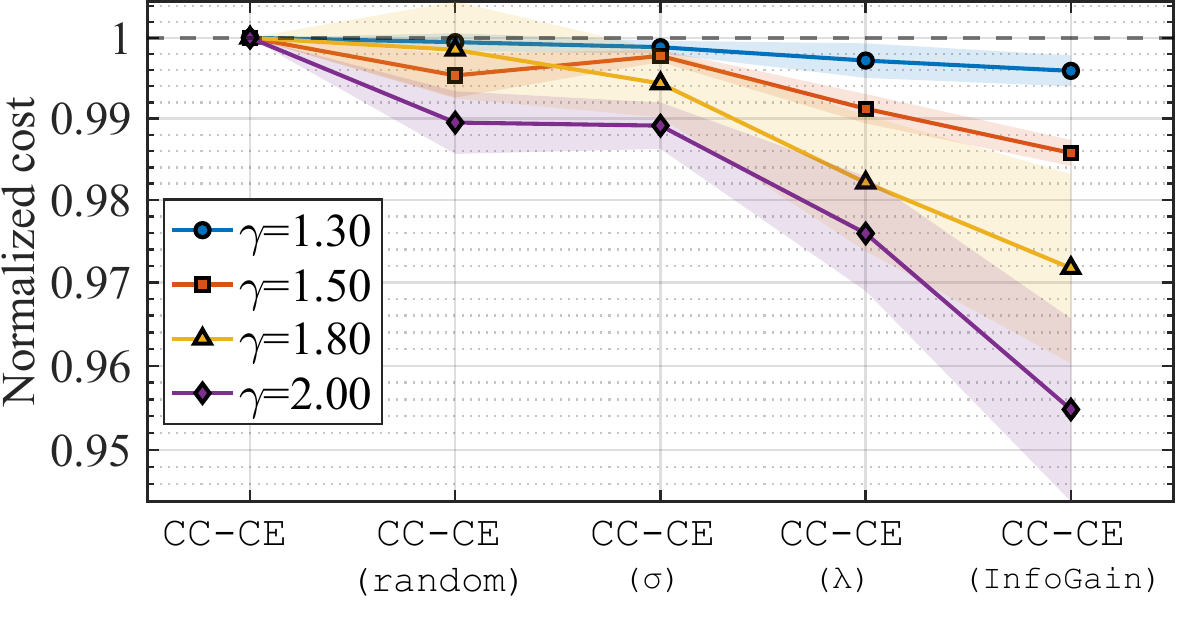}
    \caption{
    Information-acquisition performance across congestion regimes with $\alpha=0.9$.
    Each strategy selects five deviation constraints whose uncertainty margins are removed before resolving the CC-CE problem.
    Nominal costs are normalized by the baseline \texttt{CC-CE} cost for each $\gamma$.
    The proposed \texttt{InfoGain} strategy achieves the lowest normalized cost across all $\gamma$ values, with larger gains in higher-congestion regimes.}
    \label{fig:infoGainTest}
\end{figure}

\subsection{Result: Information acquisition}

We next examine whether the proposed \texttt{InfoGain} metric in \Cref{eq:infogain} identifies useful uncertainty sources to reduce.
After solving the baseline \texttt{CC-CE} problem, we remove the uncertainty margins from five selected deviation constraints and resolve the coordination problem.
We compare \texttt{random}, $\sigma$-only, $\lambda$-only, and \texttt{InfoGain} selection strategies.

As shown in \Cref{fig:infoGainTest}, \texttt{InfoGain} achieves the lowest normalized nominal cost across the tested congestion regimes.
The magnitude of improvement is scenario-dependent: the reduction is small for lower-$\gamma$ cases, but reaches approximately $4.5$\% reduction relative to the baseline \texttt{CC-CE} solution when $\gamma=2$.
The comparison with the $\sigma$-only and $\lambda$-only strategies shows that uncertainty magnitude or constraint sensitivity alone is insufficient.
Effective information acquisition requires their joint effect, together with the recommendation probability, as captured by \texttt{InfoGain}.

\section{Conclusion}

We studied uncertainty-aware coordination in multi-agent systems using chance-constrained correlated equilibria.
We showed that uncertainty contracts the feasible correlated-equilibrium region and derived sensitivity results that quantify how uncertainty affects coordination performance.
Using dual sensitivity analysis, we introduced an information-gain metric that identifies which uncertainty sources most strongly limit the achievable outcome.
We also analyzed the role of the confidence level $\alpha$, which improves robustness against deviations but may reduce coordination flexibility.

Numerical experiments confirmed this robustness--efficiency effect and showed that the proposed information-gain metric outperforms uncertainty-only and sensitivity-only information-acquisition strategies.
The present analysis relies on an agent-level Gaussian uncertainty model, which yields the quantile-based deterministic equivalent used in the local sensitivity characterization.
Future work includes extending the framework to deviation-specific, correlated, or non-Gaussian uncertainty structures and characterizing feasibility under large uncertainty levels.

\bibliographystyle{IEEEtran}
\bibliography{Ref}

\end{document}